 \documentclass[pmlr,twocolumn,10pt]{jmlr} 

\usepackage{booktabs}
\usepackage{siunitx}

\usepackage{amsmath}
\usepackage{bbm}
\usepackage{dsfont}

\theorembodyfont{\upshape}
\theoremheaderfont{\scshape}
\theorempostheader{:}
\theoremsep{\newline}

\jmlrvolume{}
\jmlryear{}
\jmlrsubmitted{}
\jmlrpublished{}
\jmlrworkshop{} 

\usepackage{algorithm}
\usepackage{algorithmic}

 \title[Anytime-valid inference in N-of-1 trials]{Anytime-valid inference in N-of-1 trials}

\author{
\Name{Ivana Malenica}
\Email{imalenica@fas.harvard.edu}
\\
\addr Harvard University, United States
\AND
\Name{Yongyi Guo}
\Email{guo98@wisc.edu}\\
\addr University of Wisconsin-Madison, United States
\AND
\Name{Kyra Gan}
\Email{kyragan@cornell.edu}\\
\addr Cornell Tech, United States
\AND
\Name{Stefan Konigorski} \Email{stefan.konigorski@hpi.de}\\
\addr Harvard University, United States; Hasso Plattner Institute, Germany
}

\begin{document}

\maketitle

\begin{abstract}
App-based N-of-1 trials offer a scalable experimental design for assessing the effects of health interventions at an individual level. Their practical success depends on the strong motivation of participants, which, in turn, translates into high adherence and reduced loss to follow-up. One way to maintain participant engagement is by sharing their interim results. Continuously testing hypotheses during a trial, known as ``peeking", can also lead to shorter, lower-risk trials by detecting strong effects early. Nevertheless, traditionally, results are only presented upon the trial's conclusion. In this work, we introduce a potential outcomes framework that permits interim peeking of the results and enables statistically valid inferences to be drawn at any point during N-of-1 trials. Our work builds on the growing literature on \textit{valid confidence sequences}, which enables anytime-valid inference with uniform type-1 error guarantees over time. We propose several causal estimands for treatment effects applicable in an N-of-1 trial and demonstrate, through empirical evaluation, that the proposed approach results in valid confidence sequences over time. We anticipate that incorporating anytime-valid inference into clinical trials can significantly enhance trial participation and empower participants.

\end{abstract}


\section{Introduction}\label{sec::intro}

The statistical inference of individual causal effects of health interventions holds great importance for
many clinical and biomedical applications. 
In particular,
understanding which treatment and dosage are most effective for a particular patient lies at the heart of personalized medicine.
To this aim, different methodologies have been proposed. One popular approach involves the collection and analysis of extensive population-level datasets with the goal of estimating effects at the individual level. With suitable covariates and a clear understanding of disease mechanisms, it becomes feasible to potentially acquire individual-level effects~\citep{Shalit2017, Bica2020, Smith2020, VerstraetePA3450, Diemert2021}. 
Nevertheless, the practicality of this approach is often limited to specific applications. For instance, in cases like cancer, personalized signatures can be derived from genetic mutations. However, more often, causal effects can only be identified within specific subgroups~\citep{vanKruijsdijk2014, Zhang2017, vanAmsterdam2022, Msaouel2022}.
As a second approach, 
stemming from recent biotechnological advancements, personalized treatments have been developed directly for selected rare target diseases in personalized drug development~\citep{Kim2019, Seydel2023}. However, such approaches are still restricted to a selected class of drug targets and limited by resources and costs.

As a third approach, experimental studies can be designed to directly evaluate and compare the effectiveness of one or multiple treatments in a given person. These so-called N-of-1 trials have been established as the gold standard for inferring individual-level effects, and different guidelines have been proposed for their standardized application ~\citep{Nikles2015, Vohra2015, Porcino2020}. More formally, N-of-1 trials are multi-crossover randomized controlled trials (RCTs) in one person, hence one or more treatments are administered over time in a predefined, potentially randomized, sequence. 
Also, if the same N-of-1 trial is performed in multiple persons in a so-called series of N-of-1 trials, the trials can be jointly analyzed to yield efficient population-level treatment effect estimates~\citep{Zucker2010}. 
In order to achieve sufficient statistical power for inference, N-of-1 trials require frequent longitudinal measurements which has hindered their application in practice. Only recently, digital tools have been developed 
for setting up and executing N-of-1 trials. This, in turn, enables scalability across various studies, patients, and providers~\citep{Taylor2018, Daskalova2020, Zenner2022, Konigorski2022}.

Digital app-based N-of-1 trials provide a straightforward and expandable means for evaluating patient outcomes, whether actively reported by patients or passively collected sensor data. However, their success still relies on the retention and high adherence of the patients to the trial. To keep the trial participants engaged, apart from developing user-friendly interfaces for the apps, it is also important to provide frequent feedback. A classical study design involves collecting all trial data, conducting analysis, and subsequently reporting the results back to the participants. However, if patients are required to provide daily outcomes in a trial that runs for weeks or even months, without any intermediate feedback, patients might lose interest. On the other hand, ``peeking'' at the intermediate results (performing a hypothesis test before the end of the trial) introduces statistical biases. Thus, a framework for valid statistical inference is required at all time points of an N-of-1 trial in order to enable intermediate analysis.

\noindent
\textbf{Contributions.} In this work, we provide a statistical framework that allows anytime-valid inference in N-of-1 trials with intermediate analysis of the results. Our contributions are to provide a (i) potential outcomes framework for N-of-1 trials, (ii) formal definition of several causal estimands of interest in N-of-1 trials, and (iii) construction of confidence sequences that allow anytime-valid inference. Finally, we validate our approach empirically in simulation studies and compare it to existing state-of-the-art approaches. 

\subsection{Related Work}\label{sec::litreview}

The field of N-of-1 trials has originated from the medical domain, and continues to be largely driven by its clinical applications. As such, the existing literature has almost exclusively focused on an applied presentation, and omitted a more formal statistical definition of the study setup and estimands of interest. There exist some recent contributions: \cite{yang2021sample} give a formalization of different potential study designs and their consequences on design parameters. \cite{Daza2018, Daza2019} and \cite{Daza2022} present a counterfactual framework for single case studies which include both observational and experimental N-of-1 settings. They define individual-level target estimands and estimators in their work, which can include time trends and carryover effects. As a difference to these previous approaches, we consider design-based estimands based on the immediate causal effects conditional on the accrued history, and focus on constructing anytime-valid confidence sequences for the proposed target parameters.

Other related work has been published in the traditional RCTs literature that is relevant to our aim of enabling anytime-valid inference. In RCTs, 
interim analyses are often performed 
to evaluate safety and side effects, as well as to assess intermediate treatment effects that might warrant early termination of the trial based on predefined criteria  for treatment inferiority or superiority. Common approaches for such interim analyses include the O’Brien-Fleming~\citep{OBrien1979}, Haybittle-Peto~\citep{Haybittle1971, Peto1976} and Pocock~\citep{Pocock1977} methods, which aim to control the overall type I error across all interim tests by adjusting the respective critical test statistic values. While the Pocock method chooses the same alpha level at all interim tests, the O’Brien-Fleming and Haybittle-Peto method require stronger evidence at earlier interim points. 
In the CONSORT extension with guidelines for harm-related stopping rules, the O’Brien-Fleming method is recommended~\citep{Ioannidis2004}. \cite{Demets1994} and others have generalized these ideas through an ``alpha-spending function" which generates critical values such that the sum of probabilities of exceeding those values across the interim tests equal the type I error rate alpha. We note that all these approaches have been designed for classical population-level RCTs, and have not been evaluated in N-of-1 trials. As only exception, \cite{Selukar2021} considers a very specific situation, when a series of N-of-1 trials is sequentially monitored, interim analyses are few and pre-defined, and only summary statistics are available of each trial. 

Third, relevant work has originated from causal inference literature on sequential designs for single time series \citep{bojinov2019time, malenica2021adaptive, ham2022design}. \cite{malenica2021adaptive} define conditional estimand classes of interest and provide inference on them in an adaptive setting of a single time series. \cite{bojinov2019time} provide a potential outcome framework for single time series and define estimands of interest from the design-based perspective. Finally, we adapt and build on work by \cite{ham2022design}, who provide design-based confidence sequences for very long time series in a setting where treatment is randomized at each time point. In this work we focus specifically on the setup of N-of-1 trials as crossover experiments in treatment blocks, with N-of-1 trial-specific estimands.

\section{Statistical Formulation}\label{sec::formulation}

\subsection{Notation and Observed Data}\label{sec::data}

We consider the trajectory of a single individual in an N-of-1 trial with $K$ treatment periods (also denoted as ``periods" or ``blocks"). Suppose for each period $k\in[K]:=\{1, \ldots, K\}$, there are $T_k$ time points. We write $(k,t)$ to indicate the $t$-th time point of the treatment period $k$. We emphasize that, despite our focus on a single N-of-1 trial, our results generalize to a series of N-of-1 trials with multiple patients.

Let $O_{k,t} := (A_{k,t}, Y_{k,t}, W_{k,t})$ denote data for a single individual at time point $(k,t)$ including the treatment, outcome, and covariates. Specifically, at each time point $(k, t)$, one assigns a binary treatment $A_{k,t} \in \mathcal{A} := \{0,1\}$ to a patient, where $A_{k,t} = 1$ denotes the treatment and $A_{k,t} = 0$ the control (or alternative treatment). In N-of-1 trials, the same treatment is assigned throughout a  period (e.g., for block $k$, control is given at all time points $(1, \ldots, T_k)$). In an extreme case, one might randomize at each follow-up time (corresponding to the chronological observation), so $T_k=1$ for all $k \in [K]$. Notice that the number of time points within a block, $T_k$, can vary by $k$, allowing for different lengths of each treatment period. 
Once the treatment is assigned for block $k$, post-treatment outcome of interest $Y_{k,t} \in \mathcal{Y}$ and possibly a vector of other time-varying covariates ${W}_{k,t} \in \mathcal{W}$ are collected. Therefore, at each time point $t$, we assign treatment $A_{k,t}$ according to the period $k$, then collect $Y_{k,t}$, followed by $W_{k,t}$.

Let ${O}_{k,1:t} = (O_{k,1}, \ldots, O_{k,t})$ denote the data collected in period $k$ up to time $t$. As each $k$ contains $T_k$ data points, we define $O_k := {O}_{k, 1:T_k} = (O_{k,1}, \ldots, O_{k,T_k})$ as the full data collected in the $k$-th treatment period. Similarly, we write
${A}_{k} = (A_{k,1}, \ldots, A_{k,T_k})$ for the full and ${A}_{k,1:t} = (A_{k,1}, \ldots, A_{k,t})$ for the cropped sequence of treatments in block $k$. To clarify, we illustrate the proposed notation with a simple example where $K=2$ and $T_1=T_2=2$. Without putting any assumptions on the design of treatment blocks (and therefore treatment assignment), all possible treatment sequences for a single individual are as follows: $\{(1,1),(0,0)\}$,  $\{(0,0),(1,1)\}$, $\{(1,1),(1,1)\}$ and $\{(0,0),(0,0)\}$. Symbolically, we write the resulting treatment sequence
as $\{(A_{1,1}, A_{1,2}), (A_{2,1}, A_{2,2})\} = 
\{{A}_{1, 1:2}, A_{2,1:2}\}$.

It follows that data collected for a whole block $k$ is then  ${O}_{k} = ({A}_{k}, {Y}_{k}, {W}_{k})$, where ${Y}_{k} = (Y_{k,1}, \ldots, Y_{k,T_k})$ and  ${W}_{k} = (W_{k,1}, \ldots, W_{k,T_k})$. When no time-varying covariates are gathered beyond the outcome of interest, we represent this as ${W}_{k} = \emptyset$. 
We emphasize that ${W}_{k}$ (and the entire history, or some function of it, for that matter) can be used to inform treatment allocation in the subsequent block, $k+1$, allowing for adaptive treatment assignment. Finally, let $O_{0}$ be the baseline covariates obtained before the start of the trial. Without loss of generality, we assume $T_k = T$ for all blocks $k\in [K]$ ($T>1$), which aligns with the design of a canonical N-of-1 trial. We, however, emphasize that all of our results generalize to the setting where the $T_k$ values vary in each treatment period.

Below, we provide some 
essential definitions regarding time and block-specific variable collections used throughout the manuscript.
For a single individual, let $\Bar{O}_{k,t} = (O_0, O_{1,1}, \ldots, O_{k,t})$ denote the full history up to time $t$ of block $k$ (including baseline covariates), while $\underline{O}_{k,t} = (O_{k,t}, \ldots, O_{K,T})$ denotes all future variables from index $(k,t)$ to $(K,T)$. Then, the full \textit{trajectory} (or a \textit{time-series}) is presented by $O^{KT}:= \bar O_{K,T}=\{O_0, O_{1,1}, \ldots, O_{K,T}\}$. Similarly, we define all past and future collections of $A$, $Y$, and $W$. For example, we write $\Bar{A}_{k,t} = (A_{1,1}, \ldots, A_{k,t})$ for the sequence of treatments until the $t$-th time point of period $k$. Lastly, let $H^A_{k,t} := \Bar{O}_{k,t-1}$ be the full history of all variables until $A_{k,t}$. It then follows that $H^Y_{k,t} := (A_{k,t}, \Bar{O}_{k,t-1})$ and $H^W_{k,t} := (A_{k,t}, Y_{k,t}, \Bar{O}_{k,t-1})$ are the full variable histories until $Y_{k,t}$ and $W_{k,t}$.

\subsection{Likelihood of the Trajectory}\label{sec::likelihood}

We let $O^{KT} \sim P_0$, where $P_0$ denotes the true probability distribution of $O^{KT}$. Throughout the remainder of the text, we use capital letters to indicate random variables, and lowercase for their realizations. We use the naught subscript to denote \textit{true} probability distributions or components thereof.
Let $\mathcal{M}$ denote the \textit{statistical model} for the probability distribution of the data, which is nonparametric, beyond possible knowledge of the treatment mechanism (i.e., known randomization probabilities). We note that the true probability distribution of the data is an element of $\mathcal{M}$, and denote $P$ as any probability distribution such that $P \in \mathcal{M}$. Suppose that $P_0$ admits a density $p_0$ w.r.t. a dominating measure $\mu$ over $\mathcal{O}$ which can be written as the product measure $\mu = \times_{k=1,t=1}^{k=K,t=T} (\mu_A \times \mu_Y \times \mu_W)$, with $\mu_A$, $\mu_Y$, and $\mu_W$ measures over $\mathcal{A}$, $\mathcal{Y}$, and $\mathcal{W}$. The likelihood of $o^{KT}$ can be factorized according to the time-ordering as 
\begin{align}\label{eqn::likelihood}
    p_0(o^{KT}) &= p_{0,O_0}(o_{0}) \prod_{k=1}^{K} \prod_{t=1}^{T} p_{0,A}(a_{k,t} \mid h^A_{k,t}) \\
    &p_{0,Y}(y_{k,t} \mid h^Y_{k,t}) p_{0,W}(w_{k,t} \mid h^W_{k,t}), \nonumber
\end{align}
where $a_{k,t} \mapsto p_{0,A}(a_{k,t} \mid h^A_{k,t})$, $y_{k,t} \mapsto p_{0,Y}(y_{k,t} \mid h^Y_{k,t})$, and $w_{k,t} \mapsto p_{0,W}(w_{k,t} \mid h^W_{k,t})$ are conditional densities w.r.t. dominating measures $\mu_A$, $\mu_Y$, $\mu_W$. 

\section{Causal Effects for N-of-1 Trials}

The main design components of an N-of-1 trial include the (1) number of blocks $K$, (2) length of each block, $T$, and (3) choice of treatment allocation for treatment periods (pre-specified or randomized, type of randomization) \citep{yang2021sample}. If the treatment sequence is pre-specified before a trial, then an individual is assigned a specific treatment sequence deterministically (e.g., control-treatment-control-treatment). This is useful when one is interested in the effect of a specific treatment sequence, or wishes to avoid randomly generating unwanted treatment sequences (e.g., giving the same treatment across consecutive periods). Alternatively, one might generate treatment sequences by randomizing treatment allocation across blocks. In this work, we focus specifically on randomized treatment sequences. 

There are numerous randomization schemes for designing an N-of-1 trial \citep{yang2021sample}. In this work, we rely on (1) \textit{pairwise randomization}, where the order of two different treatments in a consecutive pair of treatment periods is randomized; (2) \textit{restricted randomization}, where treatment is randomly assigned with the restriction that the number of treatment and control periods is approximately the same (but treatment probability is never zero) and (3) \textit{unrestricted randomization}, where treatment is randomly assigned at each period. All of the listed schemes randomize on the block-level, and assign the same treatment at all time points within a block.

\subsection{Time-Series Potential Outcomes}\label{sec::potentialY}

We define $\Bar{a}_k = ({a}_{1}, \ldots, {a}_{k})$ as the \textit{treatment path} until period $k$, 
where we remind that ${a}_{k} = a_{k,1:T}$. In a point treatment setting, the treatment path is of length 1, and each study participant has $2^1$ potential outcomes. In an N-of-1 trial, however, we follow a single individual over time and administer $K$ different treatments, each of length $T$ (or more generally, $T_k$). For a binary treatment and unrestricted randomization, we then have $2^K$ different treatment paths that could have been observed. In \tableref{table::sequences}, we include the total number of possible treatment sequences for each considered randomization scheme, at both odd and even number of treatment periods.

\begin{table}
\floatconts
  {table::sequences}
  {\caption{Number of possible treatment paths generated under different randomization schemes.
  }}
  {\begin{tabular}{lll}
  \toprule
  \bfseries Randomization & \bfseries Odd $K$ & \bfseries Even $K$\\
  \midrule
  Pairwise & $2^{\frac{K+1}{2}}$ & $2^{\frac{K}{2}}$ \\
  Restricted & $\sim 2 {K \choose (K-1)/2}$ & $ \sim {K \choose K/2}$ \\
  Unrestricted & $2^{K}$ & $2^{K}$ \\
  \bottomrule
  \end{tabular}}
\end{table}

We define $Y_{k,t}(\Bar{a}_{k,t})$ as the potential outcome at time point $(k,t)$, which may depend on the full history of assigned treatments up until $(k,t)$. Note that we don't make assumptions on carryover or other time-dependent effects. Consequently, we denote $\Bar{Y}_{k,t}(\Bar{a}_{k,t})$ as the collection of potential outcomes up until $(k,t)$: $\Bar{Y}_{k,t}(\Bar{a}_{k,t}) = (Y_{1, 1}(\Bar{a}_{1, 1}), ... , Y_{1, T}(\Bar{a}_{1, T}),..., Y_{k,1}(\Bar{a}_{k,1}),..., Y_{k,t}(\Bar{a}_{k,t}))$. 
Further, let $Y_k(\Bar{a}_k)$ denote a summary potential outcome of treatment period $k$. In particular, $Y_k(\Bar{a}_k)$ depends on the potential outcomes for all time points $t$ within block $k$, i.e., $Y_k(\Bar{a}_k) := $
$f(Y_{k,1}(\Bar{a}_{k,1}), \ldots, Y_{k,T}(\Bar{a}_{k,T}))$, where $f$ is any function that takes as input the potential outcomes in block $k$. 
Finally, let $\Bar{Y}_k(\Bar{a}_k) = (Y_1(\Bar{a}_1), \ldots, Y_k(\Bar{a}_k))$ denote the counterfactual outcomes that would have been observed over time under treatment path $\Bar{A}_k = \Bar{a}_k$. When $f$ is the average function across the respective block, we have that $\Bar{Y}_k(\Bar{a}_k) = (Y_1(\Bar{a}_1), \ldots, Y_k(\Bar{a}_k)) = (1/T \sum_{t=1}^T Y_{1,t}(\Bar{a}_{1,t}), \ldots, 1/T \sum_{t=1}^T Y_{K,t}(\Bar{a}_{K,t}))$. 

Note that we assume that future potential outcomes do not cause past treatments \citep{granger1969}. We formalize this in Assumption \ref{ass::granger}. Finally, we do not make assumptions on the dimension of any of the defined potential outcomes.

\begin{assumption}[Granger Causality]\label{ass::granger} 
Let $\Bar{W}_T = \emptyset$. For each $k \in [K]$, we have that 
\begin{align*}
    &P(A_{k,1} = a_{k,1} \mid \bar{A}_{k-1} = \bar{a}_{k-1}, \Bar{Y}_K(\Bar{a}_K)) \\
    &= P(A_{k,1} = a_{k,1}  \mid \bar{A}_{k-1} = \bar{a}_{k-1}, \Bar{Y}_{k-1}(\Bar{a}_{k-1})).
\end{align*}
\end{assumption}

To illustrate the notation, consider a simple design with unrestricted randomization, $A \in\{0,1\}$, $K=2$ and $T=2$. For the first period $k=1$, there are $2^1=2$ potential outcomes: $Y_{1,2}(0,0)$ and $Y_{1,2}(1,1)$. For the second treatment period, there are $2^2=4$ potential outcomes at the end of the treatment block: $Y_{2,2}(0,0,0,0)$, $Y_{2,2}(0,0,1,1)$, $Y_{2,2}(1,1,0,0)$, and $Y_{2,2}(1,1,1,1)$, and the total number of potential outcomes at the end of the two treatment periods under unrestricted randomization is $2 + 4 = 6$ (more generally, $2(2^K-1)$ for a trial consisting of $K$ treatment periods). The total number of treatment paths, however, is $2^K = 2^2 = 4$, out of which we observe only one. See \cite{bojinov2019time} for more examples. Let's consider the block-average function for $f$, so that each $Y_k(\Bar{a}_k)$ is the average of all potential outcomes in block $k$. For the first period, we then have that $Y_1(\Bar{a}_1) = 1/2 \sum_{t=1}^{T=2} Y_{1,t}(\Bar{a}_{1,t})$, and $Y_1(\Bar{a}_1)$ is an average of $Y_{1,1}(0)$ and $Y_{1,2}(0,0)$, or $Y_{1,1}(1)$ and $Y_{1,2}(1,1)$. Similarly, for the second treatment period $k=2$, $Y_2(\Bar{a}_2) = 1/2 \sum_{t=1}^{T=2} Y_{2,t}(\Bar{a}_{2,t})$.
Therefore, the potential outcomes that would have been observed over time are $\Bar{Y}_2(\Bar{a}_2) = (Y_1(\Bar{a}_1), Y_2(\Bar{a}_2))$.

\subsection{Target Parameter}\label{sec::target}

In an N-of-1 trial, treatments are assigned based on the current treatment period. In this section, we establish several causal estimands that might be of interest in an N-of-1 trial, using the potential outcomes as defined in Subsection \ref{sec::potentialY}. The first target parameter is the \textit{immediate causal effect} (ICE) of treatment, as opposed to control, at period $k$. It is defined as the short-term (contemporaneous) effect of administering treatments during period $k$, assessed at time $(k,T)$, right after the last treatment in period $k$, conditional on the observed past. A formal definition is provided in \definitionref{def::ice}. Another parameter of interest is the time $t$-specific ICE, which represents the causal effect of assigning treatment from time $(k,1)$ to $(k,t)$. This target parameter hints at the effect of administering treatment for $t$ time points in a treatment period. We note that ICE is a special case of the time $t$-specific ICE where $t=T$.

\begin{definition}[Immediate Causal Effect]\label{def::ice}
\begin{equation*}
    \psi_k(\Bar{a}_{k-1}) = Y_k(\Bar{a}_{k-1}, {a}_k = \textbf{1}) - Y_k(\Bar{a}_{k-1}, {a}_k = \textbf{0})
\end{equation*}
for any $k \in [K]$. $\textbf{1}$, $\textbf{0}$ are vectors of dimension $T \times 1$.
\end{definition}

\begin{definition}[Time $t$-specific ICE]\label{def::icet}
\begin{align*}
    \psi_{k,t}(\Bar{a}_{k-1}) = Y_{k,t}&(\Bar{a}_{k-1}, 
    {a}_{k,1:t} = \textbf{1}_t) \\
    &- Y_{k,t}(\Bar{a}_{k-1}, {a}_{k,1:t} = \textbf{0}_t)
\end{align*}
for any $k \in [K]$, where $\textbf{1}_t$ and $\textbf{0}_t$ are vectors of dimension $t \times 1$, and ${a}_{k,1:t} = (a_{k,1}, \ldots, a_{k,t})$.
\end{definition}

Notice that the causal estimands in both \definitionref{def::ice} and \definitionref{def::icet} are functions of the entire treatment path. As such, they include the \textit{carryover effect} from the past treatment period assignment in addition to the effect of period $k$. We also emphasize that they are \textit{data-adaptive} parameters --- the estimand changes as a function of the observed past and/or treatment path. Defining causal effects conditional on history might seem unusual from the classical causal inference perspective. However, such data-adaptive approach allows us to define causal effects (1) for long time-series, (2) with valid inference (as the central limit theorem still holds), and (3) without any additional assumptions on the time-series structure \citep{bojinov2019time}. Data-adaptive target parameters in longitudinal settings have been previously described from both super population and design-based perspectives \citep{bojinov2019time,malenica2021adaptive}.
Lastly, we define another target parameter of interest in N-of-1 trials in  \definitionref{def::aice}, the Average Immediate Causal Effect (AICE) --- the running average of treatment effects over blocks.
A similar parameter can also be defined for the average over time $t$-specific ICE.

\begin{definition}[AICE]\label{def::aice} For any $k\in[K]$,
\begin{align*}
    \psi_\mathrm{AICE}^k &= \frac{1}{k} \sum_{j=1}^k \psi_j(\Bar{a}_{j-1}). 
\end{align*}
\end{definition}

\subsection{Estimation}\label{sec::estimation}
In this paper, we focus on the design-based approach to causal inference. Within the design-based paradigm, the full set of potential outcomes is fixed and always conditioned on; as such, the only source of randomness comes from the treatment assignment. Let $\mathcal{F}_{k}$
denote the filtration which contains all observed data up to time $(k,T)$ conditional on $\Bar{Y}_k(\Bar{a}_k)$ ($\bar{O}_{k} = \{(A_{j,t},Y_{j,t}, W_{j,t})\}_{j=1,t=1}^{j=k,t=T}$), and all the potential outcomes ($\Bar{Y}_K(\Bar{a}_K)$). The filtration $\mathcal{F}_{k}$ obeys the nesting property where $\mathcal{F}_{k} \subset \mathcal{F}_{k+1}$ for all $k$. At the beginning of each period $k$, we randomly assign treatment with probability $g(A_{k,1}) = P(A_{k,1} \mid \mathcal{F}_{k-1})$. Note that $g(A_{k,1}) = g(A_{k,2}) = \ldots = g(A_{k,T})$, as the probability of treatment is the same for each time point $t$ in block $k$. By Assumption \ref{ass::granger}, it follows that $g(A_{k,1}) = P(A_{k,1} \mid \mathcal{F}_{k-1}) = P(A_{k,1} \mid \bar{A}_{k-1} = \bar{a}_{k-1}, \Bar{Y}_{k-1}(\Bar{a}_{k-1}))$. We emphasize that if treatment is assigned independently of the past, then $g(A_{k,1}) = P(A_{k,1})$. 


To estimate ICE and AICE, we focus on the time-series version of the Horvitz-Thomson and the stabilized IPTW (Hájek) estimator in this work \citep{horvitz1952,hajek1971, robins2000,hirano2003,imbens_rubin_2015}.
To enable statistical inference, we assume there is a positive probability of treatment and control at every period $k$. Formally stated in Assumption \ref{ass::positivity}, the positivity assumption excludes the pre-determined treatment periods occasionally used in N-of-1 trials.

\begin{assumption}[Positivity]\label{ass::positivity} 
For every $k \in [K]$,
    $$0 < g(A_{k,1}) < 1.$$
\end{assumption}

Under Assumption \ref{ass::positivity}, we can estimate ICE and AICE using the observed data. The IPTW estimator of $\psi_k(\Bar{a}_{k-1})$,
denoted as $\hat{\psi}_k$, is then defined as
\begin{align}\label{eqn::est_ice}
   \hat{\psi}_k := &\frac{\mathds{1}({A}_{k,1} = 1) 
   f(Y_{k,1:T})
   }{g(A_{k,1})} - \frac{\mathds{1}({A}_{k,1} = 0) f({Y}_{k, 1:T})}{1-g(A_{k,1})}. \nonumber
\end{align}
Recall $f$ is any function that takes 
data of block $k$ as input (in the estimator, observed data at time points $(k,1:T)$).
The variance estimator is defined as
\begin{align}
   \hat{\sigma}_k^2 := &\frac{\mathds{1}({A}_{k,1} = 1) f({Y}_{k, 1:T})^2}{g(A_{k,1})^2} +  \frac{\mathds{1}({A}_{k,1} = 0) f({Y}_{k, 1:T})^2}{(1-g(A_{k,1}))^2}. \nonumber
\end{align}
\lemmaref{lemma::ICEpsi} establishes that the proposed estimator is unbiased, and derives its variance over the randomization (proof in Appendix \ref{apd:proof1}).

\begin{lemma}[Properties of the ICE Estimator]\label{lemma::ICEpsi}
Under Assumption \ref{ass::positivity}, it follows that 
\begin{equation*}
    \mathbb{E}(\hat{\psi}_k - \psi_k(\Bar{a}_{k-1}) | \mathcal{F}_{k-1}) = 0
\end{equation*}
and 
\begin{align*}
    \text{Var}(\hat{\psi}_k - \psi_k(\Bar{a}_{k-1}) | \mathcal{F}_{k-1}) &\leq \mathbb{E}(\hat{\sigma}^2_k | \mathcal{F}_{k-1}).
\end{align*}
\end{lemma}

The running average immediate effect over treatment periods, i.e. the AICE, can be estimated by 
\begin{equation}\label{eqn::est_aice}
\hat{\psi}_\mathrm{AICE}^k = 1/k \sum_{j=1}^k \hat{\psi}_j.
\end{equation}
The unbiasedness of $\hat{\psi}_\mathrm{AICE}^k$ follows trivially from \lemmaref{lemma::ICEpsi}. To stabilize the variance of the IPTW, we also investigate the Hájek estimator of AICE, denoted as $\tilde{\psi}_\mathrm{AICE}^{k}$ and presented in \equationref{eqn::est_aice_hajek}. We allocate the study of the Hájek estimator of AICE to \appendixref{apd:hajek}. 
\begin{align}\label{eqn::est_aice_hajek}
\tilde{\psi}_\mathrm{AICE}^{k} := &\frac{\sum_{j=1}^k \mathds{1}({A}_{j,1} = 1) f({Y}_{j, 1:T})/g(A_{j,1})}{\sum_{j=1}^k \mathds{1}({A}_{j,1} = 1)/g(A_{j,1})} \\
&-  \frac{\sum_{j=1}^k \mathds{1}({A}_{j,1} = 0) f({Y}_{j, 1:T})/(1 - g(A_{j,1}))}{\sum_{j=1}^k\mathds{1}({A}_{j,1} = 0)/(1 - g(A_{j,1}))}, \nonumber
\end{align}
with the corresponding variance estimator:
\begin{align}\label{eqn::var_aice_hajek}
   \tilde{\sigma}_\mathrm{AICE}^2 := &\frac{\sum_{j=1}^k \mathds{1}({A}_{j,1} = 1) f({Y}_{j, 1:T})^2/g(A_{j,1})^2}{\sum_{j=1}^k \mathds{1}({A}_{j,1} = 1)/g(A_{j,1})^2} \\
   &+  \frac{\sum_{j=1}^k \mathds{1}({A}_{j,1} = 0) f(Y_{j, 1:T})^2/(1 - g(A_{j,1}))^2}{\sum_{j=1}^k \mathds{1}({A}_{j,1} = 0)/(1 - g(A_{j,1}))^2}. \nonumber
\end{align}

\section{Confidence Sequences}\label{sec::confidence_seq}

We now introduce design-based asymptotic confidence sequences for N-of-1 trials. First, we define a \textit{confidence sequence} as a sequence of confidence intervals that are uniformly valid over time (also known as \textit{anytime-valid}). We say $(I_k)_{k=1}^K$ is a valid confidence sequence with type-1 error $\alpha$ (or level $1-\alpha$) for the target parameter $(\psi^k_\mathrm{AICE})_{k=1}^K$ if for any data-dependent stopping rule at $1 \leq \tau \leq K$, 
\begin{equation}\label{eqn::alphavalid}
    P(\exists k\in \{1, \ldots, \tau\}\thinspace s.t. \thinspace\psi^k_\mathrm{AICE} \notin I_k ) \leq \alpha.
\end{equation}
Under \equationref{eqn::alphavalid}, we can perform valid inference through each $I_{k}$. Furthermore, one can terminate a trial as soon as a statistically significant effect is detected with $\psi^k_\mathrm{AICE} \notin I_k$, allowing for "peeking" during the trial duration.

Anytime-valid inference was first introduced by \cite{wald1945}. Since then, significant advancements have been made in developing confidence sequences under minimal regularity conditions \citep{howard2021time,bibaut2023nearoptimal}. Two of the key contributions include the idea of time-uniform analogues of asymptotic confidence intervals (known as \textit{asymptotic confidence sequences}), and their extension to design-based framework for anytime-valid causal inference \citep{waudbysmith2023, ham2022design}. For clarity, we provide a semi-formal definition as \definitionref{def::acs}. Informally, asymptotic confidence sequences are valid confidence sequences as the number of time-points grows. While practically this might mean we don't have valid coverage at early times, this concern is alleviated by the (i) N-of-1 design, where each period is of length $T>1$, and by the (ii) upper bound variance estimator introduced in Section \ref{sec::estimation}. 

\begin{definition}[Asymptotic Confidence Sequence]\label{def::acs}A sequence of intervals $(I_k)_{k=1}^K$ is a level $1-\alpha$ asymptotic confidence sequence for the target parameter sequence $(\psi_{\mathrm{AICE}}^k)_{k=1}^K$ if there exists a non-asymptotic confidence sequence $(I_k')_{k=1}^K$ of level $1-\alpha$ such that each interval $I_k$ shares the center with $I_k'$, and that $width(I_k) / width(I_k')\rightarrow 1\thickspace a.s.$. Moreover, we say $(I_k)_k$ has approximation rate $R$ if $width(I_k) - width(I_k') = O_{a.s.}(R)$.
\end{definition}

We now formally introduce the asymptotically valid confidence sequences for the target parameter sequence of the running average immediate causal effect, $(\psi_{\mathrm{AICE}}^k)_k$. Before stating \theoremref{thm::ConfidenceSeq} (proof in \appendixref{apd:theorem1}), we need two more assumptions. In Assumption \ref{ass::boundedpo}, we assume there is an unknown, possibly extreme constant $M$ which bounds the realized potential outcomes. As $M$ can be arbitrarily large and realizations are bounded, we consider Assumption \ref{ass::boundedpo} a mild regularity condition. Assumption \ref{ass::nonvanishingvar} concerns the behavior of the variance, and is satisfied as long as potential outcomes do not vanish over time.  

\begin{assumption}[Bounded Potential Outcomes]\label{ass::boundedpo}There exists a constant $M \in \mathbb{R}$ such that for any $k \in [K]$, and any treatment path $\bar a_k$, $|Y_k(\bar a_k)|\leq M$.
\end{assumption}

\begin{assumption}[Non-vanishing Variance]\label{ass::nonvanishingvar}Let $\tilde S_{k}:= \sum_{j=1}^k \sigma_{j}^2$, where $\sigma_{j}^2:= \frac{Y_j(\bar a_{j-1}, \mathbf{1}_T)^2}{g(A_{j, 1})} + \frac{Y_j(\bar a_{j-1}, \mathbf{0}_T)^2}{1-g(A_{j, 1})}$. Then, $\tilde S_{k}\rightarrow \infty$ as $k\rightarrow \infty$ a.s.
\end{assumption}


\begin{theorem}\label{thm::ConfidenceSeq}
Let $S_{k} = \sum_{j=1}^k\hat \sigma_j^2$. Under Assumptions \ref{ass::positivity}, \ref{ass::boundedpo} and \ref{ass::nonvanishingvar}, and for any constant $\eta>0$,
$$
\frac1k \sum_{j=1}^k\hat \psi_j\pm \frac1k \sqrt{\frac{\eta^2S_{k}+1}{\eta^2} \log\bigg(\frac{\eta^2S_{k}+1}{\alpha^2}\bigg)}
$$
forms a valid $(1-\alpha)$ asymptotic confidence sequence for the target parameter sequence $(\psi_{\mathrm{AICE}}^k)_k$, with approximation rate $o(\sqrt{\tilde S_{k}\log \tilde S_{k}}/k)$.
\end{theorem}


\begin{figure}[H]
    \centering    \includegraphics[width=0.8\linewidth]{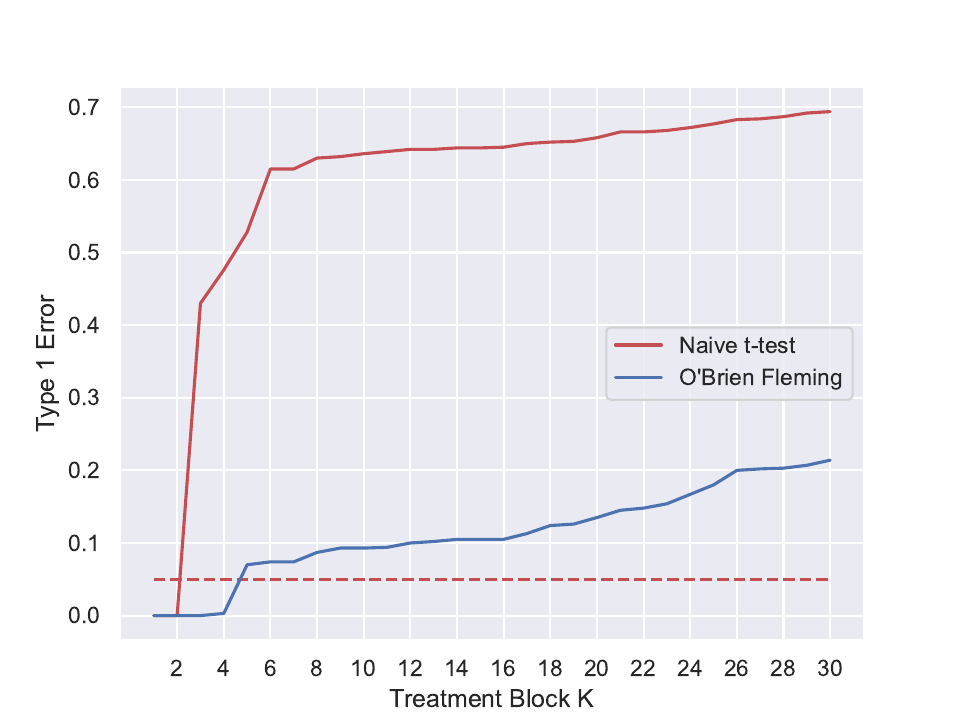}
    \caption{Empirical type I error of a naive t-test and of the O’Brien-Fleming approach 
    of 1000 N-of-1 trials in the unrestricted randomization setting. The dashed line represents $\alpha=0.05$.}
    \label{fig:baseline}
\end{figure}

\section{Experiments}\label{sec::experiments}



In the first experiment, we generate 1000 independent N-of-1 trials for an individual under the null hypothesis and illustrate the need for novel approaches to construct valid confidence sequences. In each independent trial, we set $K=30$ blocks and $T=10$ time points within each block. For simplicity, there is no treatment effect or covariates. We consider unrestricted randomization with 50\% probability for each treatment. 

First, we demonstrate the inflated type-1 error of a naive t-test (which constructs a confidence interval according to the two-sample t-test with level $\alpha = 0.05$ at each block $K$), as well as the O'Brien-Fleming approach \citep{OBrien1979}. In Figure \ref{fig:baseline}, we observe that the naive approach quickly accumulates type I errors, reaching 0.6 only after 6 treatment blocks. A naive application of the O'Brien-Fleming approach yields better performance than the simple t-test, but still result in inflated type I errors that increase over time.

The second experiment aims to visualize the performance of the proposed confidence sequences and demonstrate empirically that the sequences achieve both early stopping and time uniform coverage. For this purpose, we generate 1000 independent N-of-1 trials, each with a decreasing treatment effect size of $5 + 1/k$,
no carry-over effects, $K=100$ treatment blocks, and $T=10$ data points within each block. We consider both the unrestricted and pairwise randomization schemes, where we construct confidence sequences of AICE using our proposed IPTW (Theorem \ref{thm::ConfidenceSeq} derived from Equation \eqref{eqn::est_aice}) and stabilized IPTW estimators (Equation \eqref{eqn::est_aice_hajek}). 
In Figure \ref{fig:single_runs1} and Figure \ref{fig:single_runs2} in Appendix \ref{apd:supplfig}, 
we illustrate the confidence sequences constructed by the proposed algorithms in a single run.

\begin{figure}[H]
    \centering
    \includegraphics[width=0.85\linewidth]{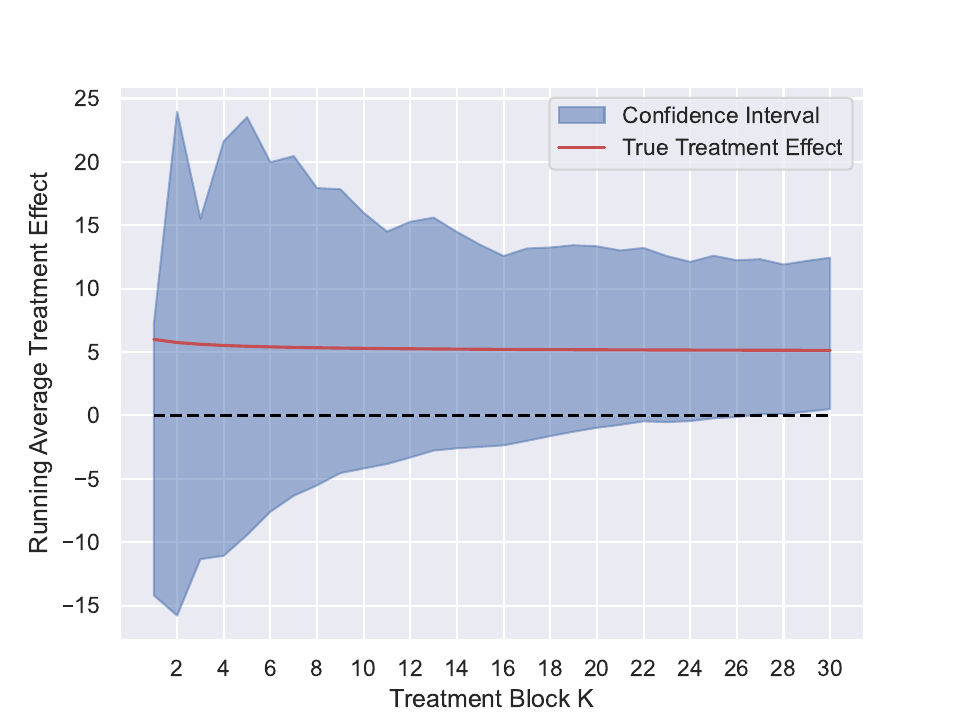}
    \includegraphics[width=0.85\linewidth]{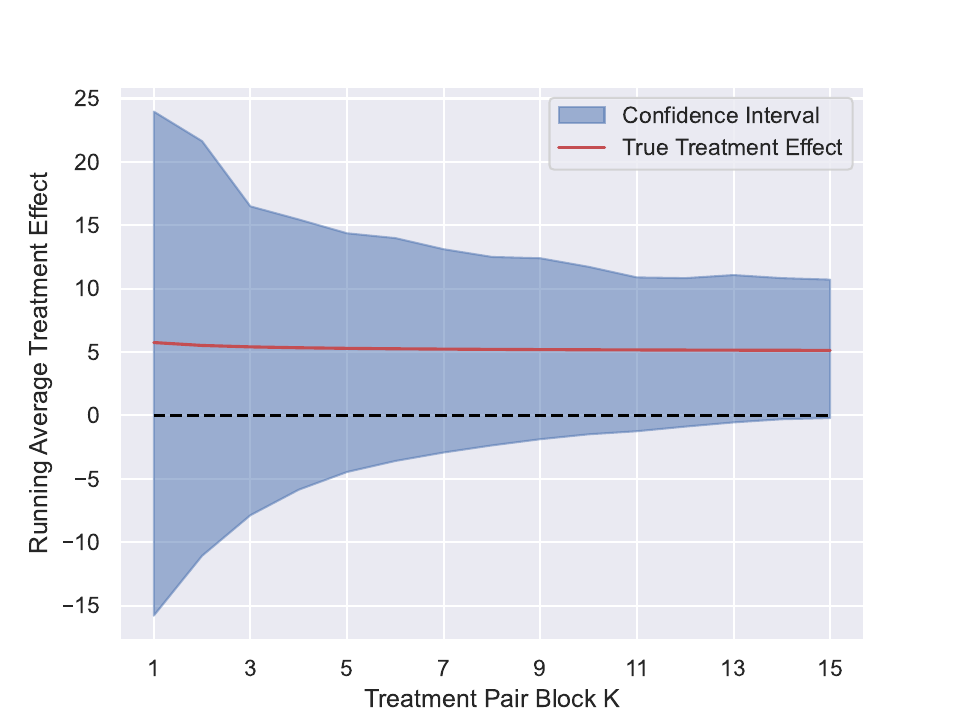}
    \caption{All-time valid confidence intervals of AICE obtained by IPTW in a single run at $\alpha=0.05$. The dashed line represents the zero (null) line. Top row: unrestricted randomization scheme. Bottom row: pairwise randomization scheme.}
    \label{fig:single_runs1}
\end{figure}

In Table \ref{tab:stopping_time_coverage}, we compute the average stopping time and time-uniform coverage proportion among the 1000 independent trials. Here the stopping time is set as the earliest block at which the confidence interval excludes 0, up to the 100th block. The time-uniform coverage refers to the proportion among the random experiments where \emph{all} confidence intervals cover the true treatment effect. 
%
%
%
%
We observe that the proposed confidence sequences enjoy high coverage and their widths are decreasing over time as expected.
The stabilized estimator shows a slightly smoother interval sequence and slightly longer average stopping times, though overall both estimators show similar empirical behavior in our considered scenarios. These insights hold for both randomization schemes.

\begin{table}[H]
    \centering
    \begin{tabular}{ccc}
    \toprule
    &   Avg. Stopping Time & Coverage \\\midrule
         IPTW & 32.50 (7.40) & 0.959\\
         S-IPTW & 31.28 (7.41) & 1.0\\
         Pair IPTW & 31.96 (2.98) & 1.0 \\
         Pair S-IPTW & 33.32 (3.96) & 1.0
         \\\bottomrule
    \end{tabular}
    \caption{Average stopping time and uniform time coverage proportion over 1000 independent trials for the proposed confidence sequences under the unrestricted and pairwise randomization setting. 
    `S-IPTW' denotes the stabilized IPTW. 
    `Pair IPTW' and `Pair S-IPTW' represent the IPTW and stabilized IPTW, respectively, in the pairwise randomization setting.
    }
    \label{tab:stopping_time_coverage}
\end{table}







\section{Discussion}\label{sec::discussion}

In this work, we provide a statistical framework that enables anytime-valid inference in N-of-1 trials for intermediate peeking at the results and analyzing them. We validate our approach in simulation studies and compare it to existing state-of-the-art approaches. The results indicate that recommended methods for population-level RCTs provide invalid confidence sequences for N-of-1 trials. Our proposed approach, however, results in valid confidence sequences.

This contribution adds to the literature on interim analysis and anytime-valid inference, with a specific focus on N-of-1 trials. Our proposed estimands allow traditional N-of-1 trials to make use of our developed methodology, and allow participants of digital N-of-1 trials to start looking at the results while the trial is ongoing. 
There is a high need for a user-friendly solution that allows peeking, and we expect that this will enable the further widespread use of N-of-1 trials.

Follow-up work can evaluate our developed methodology across a broader range of scenarios and apply it to a clinical N-of-1 trial. Further, it can be extended to population-level analyses of (adaptive) series of N-of-1 trials.



\bibliography{anytime_valid_inference_N_of_1}

\clearpage
\appendix

\section{Proof of \lemmaref{lemma::ICEpsi}}\label{apd:proof1}

The proof of \lemmaref{lemma::ICEpsi} follows from classical results in design-based causal inference and results in Appendix A of \cite{bojinov2019time}.

\vspace{3mm}
\begin{proof}[Lemma 1]
Let $g^0(A_{k,1}) = 1 - g(A_{k,1})$, $\mathds{1}_{1} = \mathds{1}({A}_{k,1} = 1)$, and $\mathds{1}_{0} = \mathds{1}({A}_{k,1} = 0)$. Then, 
\begin{align*}
\mathbb{E}(\hat{\psi}_k &\mid \mathcal{F}_{k-1}) \\
&= \mathbb{E}\left(\frac{\mathds{1}_{1} f({Y_{k,1:T}})}{g(A_{k,1})} - \frac{\mathds{1}_{0} f({Y_{k,1:T}})}{g^0(A_{k,1})} \mid \mathcal{F}_{k-1}\right) \\
&= \frac{g(A_{k,1}) Y_k(\Bar{a}_{k-1}, {a}_k = \textbf{1})}{g(A_{k,1})} \\
&\;\;\;\;\;\;\;\; - \frac{g^0(A_{k,1}) Y_k(\Bar{a}_{k-1}, {a}_k = \textbf{0})}{g^0(A_{k,1})} \\
&= Y_k(\Bar{a}_{k-1}, {a}_k = \textbf{1}) - Y_k(\Bar{a}_{k-1}, {a}_k = \textbf{0}) \\
&= \psi_k(\Bar{a}_{k-1}).
\end{align*} 
The second equality follows as $Y_k(\bar a_{k}) = f(Y_{k,1}(\bar a_{k,1}), \ldots , Y_{k,T}(\bar a_{k,T}))$ and $Y_k(\Bar{a}_{k-1}, {a}_k = \textbf{1}) = f(Y_{k,1}(\bar a_{k-1},1), \ldots , Y_{k,T}(\bar a_{k-1},\textbf{1}))$. Therefore we have that $\mathbb{E}(\hat{\psi}_k - \psi_k(\Bar{a}_{k-1}) | \mathcal{F}_{k-1}) = 0$ and $\mathbb{E}(|\hat{\psi}_k - \psi_k(\Bar{a}_{k-1})|) < \infty$. Note that the errors form a martingale difference sequence, and are uncorrelated through time. We now proceed to derive an upper bound on the variance of the proposed estimator. First, we derive the closed form expression for $\text{Var}(\hat{\psi}_k - \psi_k(\Bar{a}_{k-1}) \mid \mathcal{F}_{k-1})$. Note that
\begin{align*}
&\text{Var}(\hat{\psi}_k - \psi_k(\Bar{a}_{k-1}) \mid \mathcal{F}_{k-1}) \\
&= \mathbb{E}[(\frac{f(Y_{k,1:T})}{g(A_{k,1})}(\mathds{1}_{1} - g(A_{k,1})) \\
&\qquad\qquad\quad 
- \frac{f(Y_{k,1:T})}{g^0(A_{k,1})}(\mathds{1}_{0} - g^0(A_{k,1})))^2 \mid \mathcal{F}_{k-1}] \\
&= \mathbb{E}[\underbrace{\left(\frac{f(Y_{k,1:T})}{g(A_{k,1})}(\mathds{1}_{1} - g(A_{k,1}))\right)^2}_{\text{Term 1}} \\
&\;\;\;\;
+ \underbrace{\left(\frac{f(Y_{k,1:T})}{g^0(A_{k,1})}(\mathds{1}_{0} - g^0(A_{k,1}))\right)^2}_{\text{Term 2}} \\
&- \underbrace{2\frac{f(Y_{k,1:T})}{g(A_{k,1})}(\mathds{1}_{1} - g(A_{k,1})) \frac{f(Y_{k,1:T})}{g^0(A_{k,1})}(\mathds{1}_{0} - g^0(A_{k,1}))}_{\text{Term 3}} \mid \mathcal{F}_{k-1}].
\end{align*}

\noindent
We look at each term separately. It follows that Term 1 equals
\begin{align*}
\mathbb{E}[&\left(\frac{f(Y_{k,1:T})}{g(A_{k,1})}(\mathds{1}_{1} - g(A_{k,1}))\right)^2 \mid \mathcal{F}_{k-1}] \\
&= \frac{Y_k(\Bar{a}_{k-1}, {a}_k = \textbf{1})^2}{g(A_{k,1})^2} \mathbb{E}[(\mathds{1}_{1} - g(A_{k,1}))^2 \mid \mathcal{F}_{k-1}] \\
&= \frac{Y_k(\Bar{a}_{k-1}, {a}_k = \textbf{1})^2}{g(A_{k,1})^2} (g(A_{k,1}) - g(A_{k,1})^2).
\end{align*}

\noindent
Similarly, we have that Term 2 equals 
\begin{align*}
\mathbb{E}&\left[\left(\frac{f(Y_{k,1:T})}{g^0(A_{k,1})}(\mathds{1}_{0} - g^0(A_{k,1}))\right)^2 \mid \mathcal{F}_{k-1}\right] \\
&= \frac{Y_k(\Bar{a}_{k-1}, {a}_k = \textbf{0})^2}{g^0(A_{k,1})^2} (g^0(A_{k,1}) - g^0(A_{k,1})^2).
\end{align*}

\noindent
Finally, we obtain the Term 3, which is as follows:
\begin{align*}
&\mathbb{E}[(-2\frac{f(Y_{k,1:T})}{g(A_{k,1})}(\mathds{1}_{1} - g(A_{k,1})) \\
&\qquad\qquad\quad \frac{f(Y_{k,1:T})}{g^0(A_{k,1})}(\mathds{1}_{0} - g^0(A_{k,1}))) \mid \mathcal{F}_{k-1}] \\
&= -2 \frac{Y_k(\Bar{a}_{k-1}, {a}_k = \textbf{1})}{g(A_{k,1})} \frac{Y_k(\Bar{a}_{k-1}, {a}_k = \textbf{0})}{g^0(A_{k,1})} \\
&\qquad\qquad\quad
\mathbb{E}[(\mathds{1}_{1} - g(A_{k,1}))(\mathds{1}_{0} - g^0(A_{k,1})) \mid \mathcal{F}_{k-1}] \\
&= 2 Y_k(\Bar{a}_{k-1}, {a}_k = \textbf{1}) Y_k(\Bar{a}_{k-1}, {a}_k = \textbf{0}).
\end{align*}

\noindent
Combining all terms and using the fact that $a^2 + b^2 \geq 2ab$ (where $a=Y_k(\Bar{a}_{k-1}, {a}_k = \textbf{1})$ and $b=Y_k(\Bar{a}_{k-1}, {a}_k = \textbf{0})$, we obtain the following upper bound on the variance:
\begin{align*}
\text{Var}(&\hat{\psi}_k - \psi_k(\Bar{a}_{k-1}) \mid \mathcal{F}_{k-1}) \\
&= \frac{Y_k(\Bar{a}_{k-1}, {a}_k = \textbf{1})^2}{g(A_{k,1})^2} (g(A_{k,1}) - g(A_{k,1})^2) \\
&+ \frac{Y_k(\Bar{a}_{k-1}, {a}_k = \textbf{0})^2}{g^0(A_{k,1})^2} (g^0(A_{k,1}) - g^0(A_{k,1})^2) \\
&+ 2 Y_k(\Bar{a}_{k-1}, {a}_k = \textbf{1}) Y_k(\Bar{a}_{k-1}, {a}_k = \textbf{0}) \\
&= \frac{Y_k(\Bar{a}_{k-1}, {a}_k = \textbf{1})^2}{g(A_{k,1})} + \frac{Y_k(\Bar{a}_{k-1}, {a}_k = \textbf{0})^2}{g^0(A_{k,1})} \\
&- (Y_k(\Bar{a}_{k-1}, {a}_k = \textbf{1}) - Y_k(\Bar{a}_{k-1}, {a}_k = \textbf{0}))^2 \\
&\leq \frac{Y_k(\Bar{a}_{k-1}, {a}_k = \textbf{1})^2}{g(A_{k,1})} + \frac{Y_k(\Bar{a}_{k-1}, {a}_k = \textbf{0})^2}{g^0(A_{k,1})}.
\end{align*}
\noindent
Note that
\begin{align*}
\mathbb{E}(\hat{\sigma}^2_k &| \mathcal{F}_{k-1}) \\
&= \mathbb{E}\left(\frac{\mathds{1}_1 f({Y_{k,1:T}})^2}{g(A_{k,1})^2} + \frac{\mathds{1}_0 f({Y_{k,1:T}})^2}{g^0(A_{k,1})^2} \mid \mathcal{F}_{k-1}\right) \\
&= \frac{g(A_{k,1}) Y_k(\Bar{a}_{k-1}, {a}_k = \textbf{1})^2}{g(A_{k,1})^2} \\
&\qquad\qquad\quad + \frac{g^0(A_{k,1}) Y_k(\Bar{a}_{k-1}, {a}_k = \textbf{0})^2}{g^0(A_{k,1})^2} \\
&= \frac{Y_k(\Bar{a}_{k-1}, {a}_k = \textbf{1})^2}{g(A_{k,1})} + \frac{Y_k(\Bar{a}_{k-1}, {a}_k = \textbf{0})^2}{g^0(A_{k,1})}, 
\end{align*}
and it follows that $\text{Var}(\hat{\psi}_k - \psi_k(\Bar{a}_{k-1}) \mid \mathcal{F}_{k-1}) \leq \mathbb{E}(\hat{\sigma}^2_k | \mathcal{F}_{k-1}) = {\sigma}^2_k$.
\end{proof}

\section{Hájek estimator of the AICE}\label{apd:hajek}

In the following, we study statistical  properties of the Hájek estimator of the running average immediate effect over treatment periods, i.e. the AICE. 

\begin{lemma}[Hájek estimator of the AICE]\label{lemma::AICEpsihajek}
Under Assumption \ref{ass::positivity}, it follows that 
\begin{equation*}
    \mathbb{E}(\tilde{\psi}_\mathrm{AICE}^{k} - \psi_\mathrm{AICE}^k | \mathcal{F}_{k-1}) = 0
\end{equation*}
and 
\begin{align*}
    \text{Var}(\tilde{\psi}_\mathrm{AICE}^{k} - \psi_\mathrm{AICE}^k | \mathcal{F}_{k-1}) &\leq \mathbb{E}(\tilde{\sigma}_\mathrm{AICE}^2 | \mathcal{F}_{k-1}).
\end{align*}
\end{lemma}

\begin{proof}[Lemma 2]
Let $g^0(A_{j,1}) = 1 - g(A_{j,1})$, $\mathds{1}_{1} = \mathds{1}({A}_{j,1} = 1)$, and $\mathds{1}_{0} = \mathds{1}({A}_{j,1} = 0)$. Then, 
\begin{align*}
&\mathbb{E}(\tilde{\psi}_\mathrm{AICE}^{k} \mid \mathcal{F}_{k-1}) \\
&= \mathbb{E}(
\frac{\sum_{j=1}^k \mathds{1}_{1} f({Y_{j,1:T}})/g(A_{j,1})}{\sum_{j=1}^k\mathds{1}_{1}/g(A_{j,1})} \\
&\;\;\;\; - \frac{\sum_{j=1}^k \mathds{1}_{0} f({Y_{j,1:T}})/g^0(A_{j,1})}{\sum_{j=1}^k\mathds{1}_{0}/g^0(A_{j,1})} \mid \mathcal{F}_{k-1}) \\
&= \frac{\sum_{j=1}^k g(A_{j,1}) Y_j(\Bar{a}_{j-1},a_j = \textbf{1})/g(A_{j,1})}{\sum_{j=1}^k g(A_{j,1})/g(A_{j,1})} \\
&\;\;\;\; - \frac{\sum_{j=1}^k g^0(A_{j,1}) Y_j(\Bar{a}_{j-1},a_j = \textbf{0})/g^0(A_{j,1})}{\sum_{j=1}^k g^0(A_{j,1})/g^0(A_{j,1})} \\
&= \frac{1}{k} \sum_{j=1}^k (Y_j(\Bar{a}_{j-1}, {a}_j = \textbf{1}) - Y_j(\Bar{a}_{j-1}, {a}_j = \textbf{0})) \\
&= \psi_\mathrm{AICE}^k.
\end{align*}  
Obtaining the upper bound on the variance estimator tracks derivation already presented in the proof of \lemmaref{lemma::ICEpsi}, therefore we omit the details. It follows that  
\begin{align*}
&\text{Var}(\tilde{\psi}_\mathrm{AICE}^{k} - \psi_\mathrm{AICE}^k | \mathcal{F}_{k-1}) \\
&= \frac{\sum_{j=1}^k Y_j(\Bar{a}_{j-1}, {a}_k = \textbf{1})^2 / g(A_{j,1})}{\sum_{j=1}^k 1/ g(A_{j,1})} \\
&\;\;\;\;\;\;\;\; + \frac{\sum_{j=1}^k Y_j(\Bar{a}_{j-1}, {a}_k = \textbf{0})^2 / g^0(A_{j,1})}{\sum_{j=1}^k 1/ g^0(A_{j,1})} \\
&\;\;\;\;\;\;\;\; - (\sum_{j=1}^k Y_j(\Bar{a}_{j-1}, {a}_k = \textbf{1}) - \sum_{j=1}^k Y_j(\Bar{a}_{j-1}, {a}_k = \textbf{0}))^2 \\
&\leq \frac{\sum_{j=1}^k Y_j(\Bar{a}_{j-1}, {a}_k = \textbf{1})^2 / g(A_{j,1})}{\sum_{j=1}^k 1/ g(A_{j,1})} \\
&\;\;\;\;\;\;\;\; + \frac{\sum_{j=1}^k Y_j(\Bar{a}_{j-1}, {a}_k = \textbf{0})^2 / g^0(A_{j,1})}{\sum_{j=1}^k 1/ g^0(A_{j,1})}. 
\end{align*}
Since we have that
\begin{align*}
\mathbb{E}(&\tilde{\sigma}_\mathrm{AICE}^2 | \mathcal{F}_{k-1}) \\
&= \mathbb{E}(
\frac{\sum_{j=1}^k \mathds{1}_{1} f({Y_{j,1:T}})^2/g(A_{j,1})^2}{\sum_{j=1}^k\mathds{1}_{1}/g(A_{j,1})^2} \\
&\;\;\;\;\;\;\;\; - \frac{\sum_{j=1}^k \mathds{1}_{0} f({Y_{j,1:T}})^2/g^0(A_{j,1})^2}{\sum_{j=1}^k\mathds{1}_{0}/g^0(A_{j,1})^2} \mid \mathcal{F}_{k-1}) \\
&= \frac{\sum_{j=1}^k Y_j(\Bar{a}_{j-1}, {a}_k = \textbf{1})^2 / g(A_{j,1})}{\sum_{j=1}^k 1/ g(A_{j,1})} \\
&\;\;\;\;\;\;\;\; + \frac{\sum_{j=1}^k Y_j(\Bar{a}_{j-1}, {a}_k = \textbf{0})^2 / g^0(A_{j,1})}{\sum_{j=1}^k 1/ g^0(A_{j,1})}, 
\end{align*}
it follows that $\text{Var}(\tilde{\psi}_\mathrm{AICE}^{k} - \psi_\mathrm{AICE}^k | \mathcal{F}_{k-1}) \leq \mathbb{E}(\tilde{\sigma}_\mathrm{AICE}^2 | \mathcal{F}_{k-1})$.
\end{proof}

\section{Proof of \theoremref{thm::ConfidenceSeq}}\label{apd:theorem1}

The proof of \theoremref{thm::ConfidenceSeq} can be obtained as in \cite{ham2022design} (Theorem 5.2) and \cite{waudbysmith2023} (Theorem 2.3), adapted to the N-of-1 setting and proposed estimators. First, we state the Ville's maximal inequality in \lemmaref{lemma:ville}, as it is used in the proof \citep{ville1939}. We also use Theorem 4.4 of \cite{strassen1967}, adapted to our setting (same adaptation as in \cite{ham2022design}).

\begin{lemma}[Ville's Maximal Inequality]\label{lemma:ville}
Let $M_t$ denote a non-negative martingale with respect to a filtration $\mathcal{F}$. Then, 
\begin{equation*}
    P(\exists t \in \mathbb{N}_0 : M_t \geq \frac1\alpha) \leq \alpha M_0.
\end{equation*}
For $M_0=1$, we obtain the uniform type-1 error guarantee.
\end{lemma}

\begin{proof}[Theorem 1]
Let $X_k$ denote an i.i.d standard normal random variable. As in \cite{robbins1970}, we construct a non-negative martingale with initial value one for any $\lambda \in \mathbb{R}$ with respect to the filtration:
\begin{equation*}
    M_k(\lambda) := \exp{\left(\sum_{j=1}^k (\lambda \sigma_j X_j - \frac{\lambda^2 \sigma_j^2}{2})\right)}, 
\end{equation*}
where $\sigma_j^2$ is the true variance $\sigma_{j}^2:= \frac{Y_j(\bar a_{j-1}, \mathbf{1}_T)^2}{g(A_{j, 1})} + \frac{Y_j(\bar a_{j-1}, \mathbf{0}_T)^2}{1-g(A_{j, 1})}$. By \cite{robbins1970}, for any probability distribution $F(\lambda)$ on $\mathbb{R}$, $\int_{\lambda} M_k(\lambda) dF(\lambda)$ is still a non-negative martingale with initial value one. For a normal mixing distribution with mean $\lambda$ and variance $\eta^2$, the resulting martingale $M_k$ is written as
\begin{align*}
\frac{1}{\eta \sqrt{2\pi}} \int_{\lambda} \exp{\left(\sum_{j=1}^k (\lambda \sigma_j X_j - \frac{\lambda^2 \sigma_j^2}{2})\right)} \exp{\left(\frac{-\lambda^2}{2 \eta^2}\right)} d\lambda.
\end{align*}
Let $Z_k = \sum_{j=1}^k \sigma_j X_j$ and $\Bar{\sigma}^2_k = 1/k \sum_{j=1}^k \sigma^2_j$. Then, 
\begin{align*}
M_k &= \frac{1}{\eta \sqrt{2\pi}} \int_{\lambda} \exp{\left(\lambda Z_k - \frac{k \lambda^2 \Bar{\sigma}^2_k}{2}\right)} \exp{\left(\frac{-\lambda^2}{2 \eta^2}\right)} d\lambda \\
&= \frac{1}{\eta \sqrt{2\pi}} \int_{\lambda} \exp{\left(\frac{2\eta^2\lambda Z_k -\lambda^2(1+k\eta^2\Bar{\sigma}^2_k)}{2\eta^2}\right)} d\lambda \\
&= \frac{1}{\eta \sqrt{2\pi}} \int_{\lambda} \exp{\left(\frac{2\lambda b - \lambda^2 a}{2\eta^2} \right)}, 
\end{align*}
where $a=k\eta^2 \Bar{\sigma}^2_k + 1$ and $b=\eta^2Z_k$. We proceed to complete the square by adding and subtracting $(b/a)^2$:
\begin{align*}
\exp{\left(\frac{2\lambda b - \lambda^2 a}{2\eta^2} \right)} &=  \exp{\left(\frac{b}{\eta^2}\lambda - \frac{a}{2\eta^2}\lambda^2 \right)} \\
&= \exp{\left(
- \frac{a}{2\eta^2}(\lambda^2 - \frac{2b}{a}\lambda)\right)} \\
&= \exp{\left( - \frac{a}{2\eta^2}(\lambda - \frac{b}{a})^2 + \left(\frac{b^2}{2a\eta^2}\right) \right)}.
\end{align*}
Substituting back into $M_k$, and replacing $a$ and $b$ for their values, we have that
\begin{align*}
M_k &=  \frac{1}{\sqrt{2\pi\eta^2/a}} \int_{\lambda} \exp{\left( - \frac{a}{2\eta^2}(\lambda - \frac{b}{a})^2\right)} d\lambda \\
&\;\;\;\;\;\;\;\;\;\;\;\;\;\;\;\; \frac{1}{\sqrt{a}} \exp{\left(\frac{b^2}{2a\eta^2} \right)} \\
&= \frac{1}{\sqrt{a}} \exp{\left(\frac{b^2}{2a\eta^2} \right)} \\
&= \frac{1}{\sqrt{k\eta^2 \Bar{\sigma}^2_k + 1}} \exp{\left(\frac{\eta^2Z_k^2}{2(k\eta^2 \Bar{\sigma}^2_k +1) } \right)}.
\end{align*}
By \cite{robbins1970}, we know that $M_k$ is a non-negative martingale. Applying \lemmaref{lemma:ville}, it follows that
\begin{align*}
P(\forall k \geq 1, M_k < \frac{1}{\alpha}) \geq 1-\alpha.
\end{align*}
By basic algebra manipulation, we can show that $\frac{1}{\sqrt{k\eta^2 \Bar{\sigma}^2_k + 1}} \exp{\left(\frac{\eta^2Z_k^2}{2(k\eta^2 \Bar{\sigma}^2_k +1) } \right)} < \frac{1}{\alpha}$ simplifies to 
\begin{align}\label{eqn::villeapply}
P\left(\forall k \geq 1,  
\left|\frac{1}{k}Z_k\right| < C_k\right) \geq 1-\alpha, 
\end{align}
where 
$C_k = \sqrt{\frac{2(k\eta^2\Bar{\sigma}^2_k +1)}{k^2\eta^2} \log{\left( \frac{\sqrt{k \eta^2 \Bar{\sigma}_k^2 + 1}}{\alpha}\right)}}$.

Let $u_k = \hat{\psi}_k - \psi_k(\Bar{a}_{k-1})$. In \appendixref{apd:proof1}, we discuss that $\mathbb{E}(\hat{\psi}_k - \psi_k(\Bar{a}_{k-1}) | \mathcal{F}_{k-1}) = 0$ and $\mathbb{E}(|\hat{\psi}_k - \psi_k(\Bar{a}_{k-1})|) < \infty$, so $(u_k)_k$ is a martingale difference sequence with respect to $\mathcal{F}_{k-1}$ (and thus, uncorrelated through time). We note that the same applies for martingale difference sequence formed by the Hájek estimator, as shown in \lemmaref{lemma::AICEpsihajek}. We proceed to utilize the strong approximation theorem from \cite{strassen1967} (Theorem 4.4), which applies to martingale difference sequences of the form $\mathbb{E}(X_n | \sigma(X_1, \ldots, X_{n-1})) = 0$ (notation as in \cite{strassen1967}). In the Appendix D,  \cite{ham2022design} edit Theorem 4.4 in order to accommodate martingale difference sequences we also have, of the form $u_k$. The edited formulation of Theorem 4.4 follows due to Assumption \ref{ass::boundedpo}, and results in the following approximation:
\begin{align*}
\frac{1}{k}\sum_{j=1}^k u_j = \frac{1}{k}\sum_{j=1}^k \sigma_jX_j + o\left(\frac{\tilde{S}_k^{3/8}\log{\tilde{S}_k}}{k} \right) a.s.
\end{align*}
where we remind that $\tilde{S}_k = \sum_{j=1}^k \sigma_j^2$. Substituting back into \equationref{eqn::villeapply}, we have that
\begin{align}\label{eqn::nacsw}
P(\forall k \geq 1,  
\left|\frac{1}{k} \sum_{j=1}^k u_j\right| < C_k + o\left(\frac{\tilde{S}_k^{3/8}\log{\tilde{S}_k}}{k}\right)) \geq 1-\alpha. 
\end{align}
Note that \equationref{eqn::nacsw} is the non-asymptotic confidence width, denoted as $(I_k')_{k=1}^K$ in \definitionref{def::acs}. By Assumption \ref{ass::nonvanishingvar}, the asymptotic confidence width is 
\begin{equation}\label{eqn::acsw}
P\left(\forall k \geq 1,  
\left|\frac{1}{k} \sum_{j=1}^k u_j\right| < C_k\right) \geq 1-\alpha.    
\end{equation}
Also by Assumption \ref{ass::nonvanishingvar}, the $(1-\alpha)$ asymptotic confidence sequence for the target parameter sequence $(\psi_{\mathrm{AICE}}^k)_k$ is then
\begin{equation}\label{eqn::acs}
    \frac1k \sum_{j=1}^k\hat \psi_j \pm \sqrt{\frac{2(k\eta^2\Bar{\sigma}^2_k +1)}{k^2\eta^2} \log{\left( \frac{\sqrt{k \eta^2 \Bar{\sigma}_k^2 + 1}}{\alpha}\right)}}.
\end{equation}
Finally, note that the $(1-\alpha)$ asymptotic confidence sequence in \equationref{eqn::acs} relies on the true variance. In step 3 of Appendix D in \cite{ham2022design} and Appendix A.2 of \cite{waudbysmith2023}, they show that under further assumption of $\tilde{\sigma}^2_k \xrightarrow{a.s.}\Bar{\sigma}_k^2$, with $\tilde{\sigma}^2_k = 1/k \sum_{j=1}^k \hat{\sigma}_j^2$, we have that $C_k = \sqrt{\frac{2(k\eta^2\tilde{\sigma}^2_k +1)}{k^2\eta^2} \log{\left( \frac{\sqrt{k \eta^2 \tilde{\sigma}_k^2 + 1}}{\alpha}\right)}}$. We include the argument for clarity in what follows. Let $\Bar{\sigma}_k^2 - \tilde{\sigma}^2_k = o(\Bar{\sigma}_k^2)$. Then we have that
\begin{align*}
&\sqrt{\frac{2(k\eta^2\Bar{\sigma}^2_k +1)}{k^2\eta^2} \log{\left( \frac{\sqrt{k \eta^2 \Bar{\sigma}_k^2 + 1}}{\alpha}\right)}} \\
&= \sqrt{\frac{2(k\eta^2 (\tilde{\sigma}^2_k + o(\Bar{\sigma}_k^2)) +1)}{k^2\eta^2} \log{\left( \frac{\sqrt{k \eta^2 (\tilde{\sigma}^2_k + o(\Bar{\sigma}_k^2)) + 1}}{\alpha}\right)}} \\
&= \sqrt{\frac{k\eta^2 (\tilde{\sigma}^2_k + o(\Bar{\sigma}_k^2)) +1)}{k^2\eta^2} \log{\left( \frac{k \eta^2 (\tilde{\sigma}^2_k + o(\Bar{\sigma}_k^2)) + 1}{\alpha^2}\right)}} \\
&= \sqrt{\underbrace{\left(\frac{k\eta^2\tilde{\sigma}^2_k +1}{k^2\eta^2} + o(\Bar{\sigma}_k^2/k)\right)}_{\text{Term 1}} 
\underbrace{\log{\left( \frac{ k \eta^2 \tilde{\sigma}^2_k + o(k\Bar{\sigma}_k^2) + 1}{\alpha^2}\right)}}_{\text{Term 2}}}.
\end{align*}
We focus on Term 2 in the rest of the proof.
\begin{align*}
\log&\left( \frac{ k \eta^2 \tilde{\sigma}^2_k + o(k\Bar{\sigma}_k^2) + 1}{\alpha^2}\right) \\
&\quad\quad\quad = \log\left( \frac{ k \eta^2 \tilde{\sigma}^2_k + 1}{\alpha^2}\right) + \log{(1+o(1))} \\
&\quad\quad\quad = \log\left( \frac{ k \eta^2 \tilde{\sigma}^2_k + 1}{\alpha^2}\right) + o(1).
\end{align*}

Therefore, valid $(1-\alpha)$ asymptotic confidence sequence for $(\psi_{\mathrm{AICE}}^k)_k$ is
$$
\frac1k \sum_{j=1}^k\hat \psi_j\pm \frac1k \sqrt{\frac{\eta^2\sum_{j=1}^k\hat \sigma_j^2+1}{\eta^2} \log\bigg(\frac{\eta^2\sum_{j=1}^k\hat \sigma_j^2+1}{\alpha^2}\bigg)}.
$$

\end{proof}

\section{Supplementary Figures}\label{apd:supplfig}

\begin{figure}[H]
    \centering
    \includegraphics[width=0.49\linewidth]{images/IPTW_single.pdf}
    \includegraphics[width=0.49\linewidth]{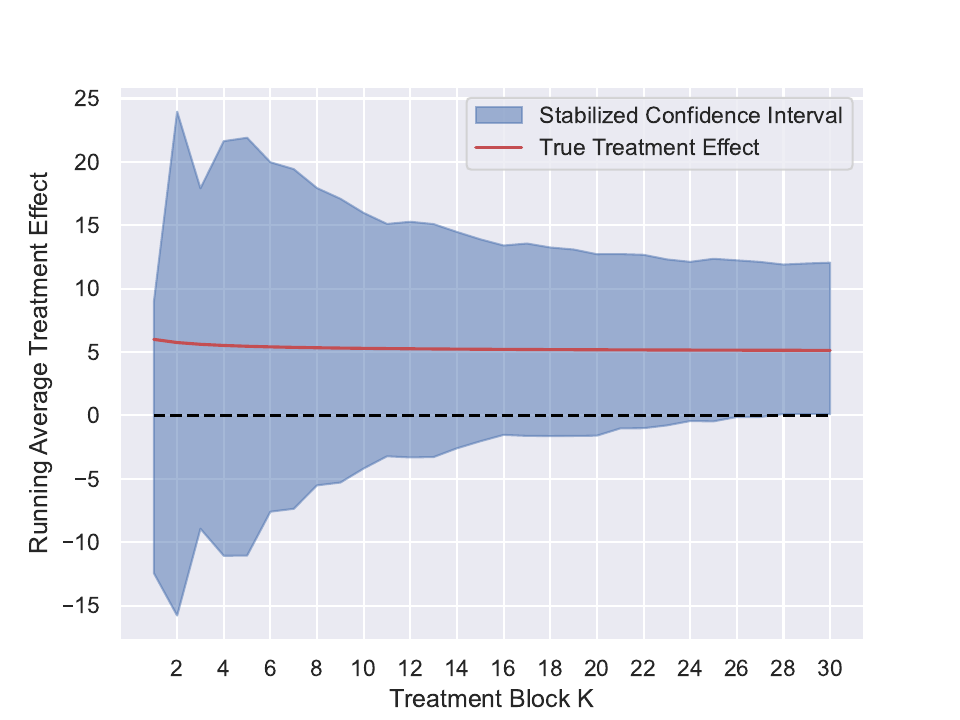}
    \includegraphics[width=0.49\linewidth]{images/pair_IPTW_single.pdf}
    \includegraphics[width=0.49\linewidth]{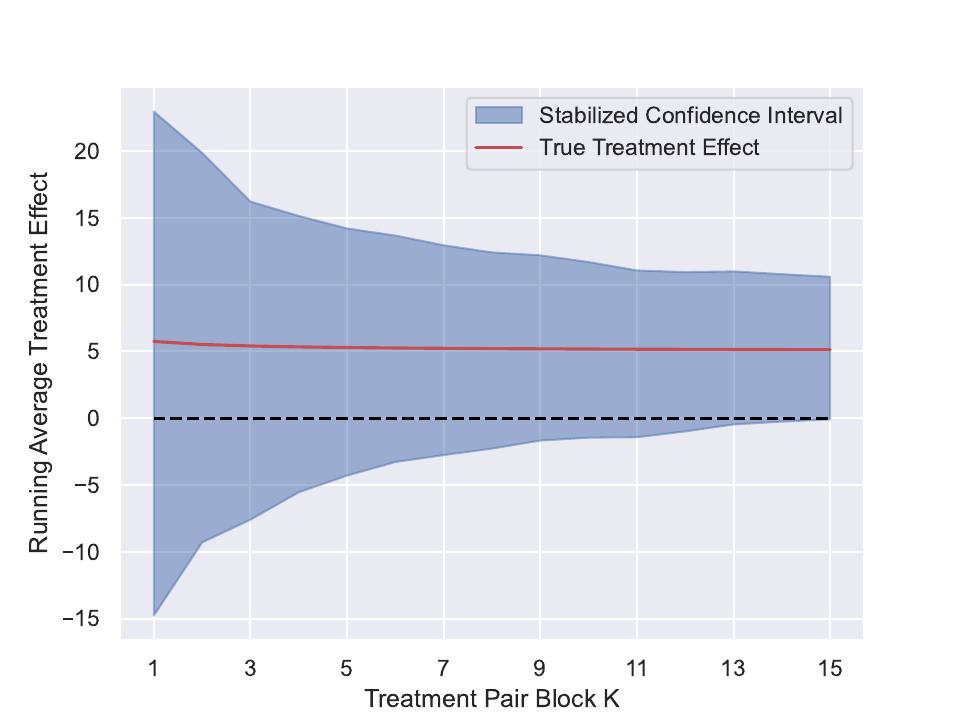}
    \caption{All-time valid confidence intervals of AICE obtained by IPTW (left) and stabilized IPTW (right), respectively, in a single run at $\alpha=0.05$. The dashed line represents the zero (null) line. Top row: unrestricted randomization scheme. Bottom row: pairwise randomization scheme.}
    \label{fig:single_runs2}
\end{figure}

\end{document}